\documentclass[sigconf]{acmart}
\usepackage{balance}
\def\BibTeX{{\rm B\kern-.05em{\sc i\kern-.025em b}\kern-.08emT\kern-.1667em\lower.7ex\hbox{E}\kern-.125emX}}

\usepackage{booktabs} % For formal tables

\usepackage{algorithm}
\usepackage{tikz}
\usepackage{algpseudocode}
\usepackage{amssymb}
\usepackage{amsmath}
\usepackage{algorithmicx}
\usepackage{algpseudocode}
\usepackage{comment}
\usepackage{color}
\usepackage{url}
\usepackage{enumitem}
\usepackage{array}
\usepackage{multirow}

\usepackage[utf8]{inputenc}
\usepackage{dsfont}
\usepackage[T1]{fontenc}

\theoremstyle{remark}
\newtheorem*{remark}{Remark}

\DeclareMathOperator{\Tr}{tr}
\DeclareMathOperator{\Rank}{Rank}

\def \bA {\mathbf{A}}

\def \bbE {\mathbb{E}}

\def \bbS {\mathbb{S}}
\def \bbB {\mathbb{B}}

\copyrightyear{2019} 
\acmYear{2019} 
\setcopyright{acmcopyright}
\acmConference[SIGIR '19]{Proceedings of the 42nd International ACM SIGIR Conference on Research and Development in Information Retrieval}{July 21--25, 2019}{Paris, France}
\acmBooktitle{Proceedings of the 42nd International ACM SIGIR Conference on Research and Development in Information Retrieval (SIGIR '19), July 21--25, 2019, Paris, France}
\acmPrice{15.00}
\acmDOI{10.1145/3331184.3331264}
\acmISBN{978-1-4503-6172-9/19/07}

\begin{document}

\title{Variance Reduction in Gradient Exploration for Online Learning to Rank }

\author{Huazheng Wang, Sonwoo Kim, Eric McCord-Snook, Qingyun Wu, Hongning Wang}
\affiliation{%
  \institution{ Department of Computer Science, University of Virginia }
  \streetaddress{85 Engineer's Way}
  \city{Charlottesville} 
  \state{VA 22904, USA} 
}
\email{{hw7ww, sak2km, esm7ky, qw2ky, hw5x}@virginia.edu}

%
% By default, the full list of authors will be used in the page headers. Often, this list is too long, and will overlap
% other information printed in the page headers. This command allows the author to define a more concise list
% of authors' names for this purpose.
\renewcommand{\shortauthors}{Wang et al.}

\begin{abstract}
Online Learning to Rank (OL2R) algorithms learn from implicit user feedback on the fly. The key to such algorithms is an unbiased estimate of gradient, which is often (trivially) achieved by uniformly sampling from the entire parameter space. Unfortunately, this leads to \emph{high-variance} in gradient estimation, resulting in high regret during model updates, especially when the dimension of the parameter space is large. 

%by exploring the parameter space. These  first propose one or more candidate exploratory directions as estimated gradients, and then update the current ranker if the proposed one is preferred by the user based on the results of an interleaved test. However, in prior works, the estimated gradients have traditionally been randomly sampled from the whole vector space to prevent bias, which results in . %The high dimensionality of the data 

%without considering whether the gradients can be \emph{distinguished} by interleaved test in current query. This is inefficient as the proposed rankers are likely to rank the candidate documents in the same order as current ranker for interleaving, and no click feedback can distinguish their ranking quality.

In this work, we focus on Dueling Bandit Gradient Descent (DBGD) based OL2R algorithms, which constitute a major endeavor in this direction of research. In particular, we aim at reducing the variance of gradient estimation in DBGD-type OL2R algorithms. %to improve the efficiency of gradient exploration. 
We project the selected updating direction (i.e., the winning direction) into a space spanned by the feature vectors from examined documents under the current query (termed the ``document space'' for short), \emph{after} an interleaved test. Our key insight is that the result of an interleaved test is solely governed by a user's relevance evaluation over the \emph{examined} documents. Hence, the true gradient introduced by this test is only reflected in the constructed document space, and components of the proposed gradient which are orthogonal to the document space can be safely removed, for variance reduction purpose. We prove that this projected gradient is still an unbiased estimation of the true gradient, and show that this lower-variance gradient estimation results in significant regret reduction. Our proposed method is compatible with all existing DBGD-type OL2R algorithms which rank documents using a linear model. Extensive experimental comparisons with several best-performing DBGD-type OL2R algorithms have confirmed the effectiveness of our proposed method in reducing the variance of gradient estimation and improving overall ranking performance.
\end{abstract}




\maketitle

\section{Introduction}
%Learning to rank is one of the most important research directions in information retrieval \cite{liu2009learning}. It provides principled solutions for numerous IR applications, such as document ranking \cite{burges2005learning,chapelle2011yahoo}, query suggestion \cite{shokouhi2013learning,ozertem2012learning}, and content recommendation \cite{rendle2009learning,yang2011collaborative}. The key advantage of learning to rank solutions is to directly optimize a function of hundreds or even thousands of ranking features with respect to ranking-based metrics, such as normalized discounted cumulative gain (NDCG) and mean average precision (MAP). Traditionally, learning to rank algorithms are optimized \emph{offline}, utilizing a training dataset with explicit relevance judgments. However, there are several evident limitations of this kind of offline learning to rank solutions: 1) manually annotating a reliable dataset for a large-scale information retrieval system is expensive and time-consuming; 2) annotator's relevance judgments may not necessarily align with actual users' result preferences \cite{radlinski2008does}; and 3)  collections of personal search history for personalized learning to rank training \cite{wang2016learning} also raises privacy concerns. 

Online Learning to Rank (OL2R) \cite{grotov2016online} is a family of online learning solutions, which exploit implicit feedback from users to directly optimize parameterized rankers on the fly. It has drawn increasing attention in research community in recent years due to its advantages over classical offline learning to rank algorithms \cite{liu2009learning}. First, it avoids the expensive and time consuming process of offline result relevance annotation. Second, as it directly learns from user feedback, it optimizes the ranking results to best reflect current user preferences \cite{radlinski2008does}. Third, because the model is updated on the fly, there is no need to store user click history offline, which alleviates many privacy concerns \cite{wang2016learning}.
 
One strain of OL2R algorithms, represented by Dueling Bandit Gradient Descent (DBGD) \cite{yue2009interactively}, optimize a linear scoring function by exploring the parameter space via interleaved test. Algorithms of this type first propose an exploratory direction as a tentative model update direction, and then update the current ranker if the proposed direction provides better ranking utility. In practice, result utility is usually inferred from user clicks on an interleaved list of ranking results from each ranker \cite{yue2010learning}. The key technical insight of DBGD-type algorithms is that the expectation of selected directions is an \emph{unbiased} estimate of true gradient of the unknown loss function for ranking \cite{flaxman2005online}. As a result, DBGD is essentially a stochastic online gradient descent algorithm. However, because the exploration directions are uniformly sampled from the entire parameter space, when the dimensionality of the space is high (which is usually the case in practice), the \emph{variance} in gradient estimation becomes large. This directly slows down the learning convergence of the algorithm and inevitably increases sample complexity.

Recently, several follow-up works have realized this deficiency of gradient exploration in DBGD, and propose various types of solutions to improve its learning efficiency. One type of studies explore multiple random directions in each iteration of model update. Unbiased estimate of gradient is maintained in this type of revisions of DBGD, as the directions are still uniformly sampled. Model estimation variance is expected to be reduced by testing more exploratory directions; but, in practice, as the users would only examine a finite number of documents under each query (e.g., due to position bias \cite{joachims2017accurately}), the sensitivity of interleaved test drops as a result of more exploratory rankers having to be tested at once. This unfortunately introduces additional variance in model estimation. Another type of research constrains the sampling space for gradient exploration \cite{hofmann2013reusing,oosterhuis2017balancing,wang2018efficient}. However, this line of solutions cannot guarantee the estimated gradient remains unbiased, and thus face high risk of converging towards a sub-optimal solution. 

Although empirically effective, previous OL2R solutions neglect an important property of click-based result utility evaluation: users only perceive utility from the documents that they actually examine. As a result, the \emph{true} gradient is only revealed by features playing an essential role in ranking those examined documents under this query. Here we define essential features in ranking a particular set of documents as those features with non-zero variance among the documents. Assume in an interleaved test, one ranking feature takes a constant value in all examined documents under this query, such that it has no effect in differentiating the quality of those documents. Then, the proposed exploratory direction's contribution to the ranker update on this particular dimension cannot be justified by this test result. Random gradient exploration hence introduces an arbitrary update on this dimension, which inevitably leads to high estimation variance over time. This example can be generalized to situations where multiple (even correlated) features have no effect in differentiating the utility of examined documents in the result of an interleaved test. Because in practice users usually only examine a handful of documents under each query \cite{joachims2017accurately,craswell2008experimental}, but each document consists of hundreds or even thousands of ranking features, the variance introduced by random exploration on those non-essential features could be considerably large.      
%user's click behaviour is a form of implicit feedback about the quality of examined documents, while non-examined documents do not contribute to the loss function. 

The above analysis suggests that an interleaved test only reveals the projection of true gradient in the spanned space of examined documents under a test query (termed the ``document space'' in this paper). With this as our motivation, we decide to project the winning direction back into the document space so as to reduce the variance introduced by random gradient exploration. We construct the document space from inferred users' result examinations \cite{craswell2008experimental}, which are not observable in the user response but can be statistically modeled. Because this projection is independent from how the proposal directions are created, this solution can be directly applied to any DBGD-type OL2R algorithm. We theoretically prove that the projected direction is still an unbiased estimate of the true gradient, i.e., model convergence is guaranteed, and also prove the reduced variance directly leads to considerable regret reduction in online model update. We compare the proposed method with several best-performing DBGD-type OL2R algorithms on a collection of large-scale learning to rank datasets and confirmed the effectiveness of our proposed solution.
\section{Related Work}
%Learning to rank has been studied as a core research problem in Information Retrieval. Many offline learning to rank algorithms have been developed and deployed in practical systems \cite{liu2009learning, burges2010ranknet}. However, the offline learning scheme requires high-quality labeled datasets \cite{liu2007letor, chapelle2011yahoo}. Serious concerns have been raised on the algorithms' dependency of labeled dataset, for example collecting annotations is time-consuming and expensive; the annotation may not accurately reflect user satisfaction \cite{sanderson2010test}; and privacy concerns for collecting personal search history \cite{wang2016learning}. These drawbacks directly motive the recent development in Online Learning to Rank (OL2R) research.

One key family of OL2R methods root in Dueling Bandit Gradient Descent (DBGD) \cite{yue2009interactively}, which uses online gradient descent to solve a bandit convex optimization problem \cite{flaxman2005online}. In each iteration, DBGD uniformly samples a random direction from the entire parameter space to create an exploratory ranker, and uses an interleaved test \cite{radlinski2008does} to compare the current ranker with the exploratory one. If the exploratory ranker is preferred, the proposed direction is used as the gradient to update the model. This procedure yields an unbiased estimate of true gradient \cite{yue2012k}. However, the variance of DBGD's gradient estimation is high due to the nature of uniform exploration of the entire parameter space, which limits its learning efficiency.

Recently, attempts have been made to improve the learning efficiency of DBGD-type algorithms. Schuth et al. \cite{schuth2016multileave} proposed a Multileave Gradient Descent (MGD) algorithm to explore multiple stochastic directions in each iteration with multi-interleaving comparison \cite{schuth2014multileaved}. Zhao and King \cite{zhao2016constructing} developed a Dual-Point Dueling Bandit Gradient Descent algorithm to sample two stochastic vectors with opposite directions as the candidate gradients. The basic idea of this line of solutions is to test more exploratory directions at once so as to obtain the true gradient estimate sooner. However, their gradient exploration is still within the entire feature space. As users often only examine a small number of documents under each query, the sensitivity of interleaved test drops due to more exploratory rankers need to be tested. In a different direction of solutions, researchers proposed to constrain the sampling space for gradient exploration. Hofmann et al. chose to filter the stochastic directions by historical comparisons before an interleaved test \cite{hofmann2013reusing}. Oosterhuis et. al \cite{oosterhuis2017balancing} proposed to explore gradients in a subspace constructed by a set of pre-selected reference documents from an offline training corpus. Wang et al. \cite{wang2018efficient} proposed to use historical interactions to avoid repeatedly exploring less promising directions, which also reduces gradient exploration to a subspace. However, the variance of gradient exploration is reduced at a cost of introducing bias into gradient approximation, so that such algorithms have a risk of converging to sub-optimal results. 

Our solution falls into this second category of variance reduction approaches for DBGD-type algorithms. Distinct from previous attempts to restrict gradient exploration before an interleaved test, we instead modify the selected direction after the test. As users' result examination is affected by the ranked results, which are in turn determined by the proposed exploratory directions, restricting the exploration space before the interleaved test potentially introduces bias in the subsequent interleaved test and model update. Our solution is based on the insight that only the projected true gradient in the document space can be revealed by an interleaved test. Hence, we decide to project the selected direction after each interleaved test, and thus guarantee an unbiased estimate of true gradient. Since the document space is expected to be smaller than the entire parameter space (as it is constructed only by the examined documents), the projected gradient enjoys low variance and leads to faster model convergence in online update. 

There are also other parallel lines of OL2R algorithms that do not explore the gradient space the way DBGD-type algorithms do, but directly optimize the ranking model from click feedback. Kveton et al. \cite{kveton2015cascading} proposed Cascading Bandits to learn from users' click behaviour, where skipped documents are assumed to be less attractive than later clicked ones. This model is then extended to the dependent click model \cite{katariya2016dcm} to support multiple clicks in one query, and further studied for general stochastic click models \cite{zoghi2017online}. However, these algorithms estimate a separate model for each query and do not share estimation across queries, which lead to slow convergence. Oosterhuis et al. \cite{Oosterhuis_2018} proposed a Pairwise Differentiable Gradient Descent (PDGD) algorithm that constructs gradients from pairwise result comparisons to update the model, and can be used to optimize neural network models. We should note that our solution is not compatible nor directly comparable with these non-DBGD algorithms, as there is no gradient exploration in these algorithms and our proposed gradient project does not apply. 
\section{Method}

In this section we describe our proposed document space gradient projection method for online learning to rank.
We first describe the problem setup in Section \ref{sec:setup}. Then we present Document Space Projected Dueling Bandit Gradient Descent (DBGD-DSP) algorithm as an example of our proposed general solution in Section \ref{sec:alg}. Our gradient projection method is independent from how the exploratory gradient is proposed, and thus can be directly applied to any existing DBGD-type OL2R algorithm \footnote{In the following discussions, we will use "DBGD-type OL2R algorithm" and "OL2R algorithm" interchangeably, as the focus of this paper is improving the efficiency of DBGD-type OL2R algorithms.
} to reduce its variance of gradient estimation.  We rigorously prove the unbiasedness of our gradient estimation in Section \ref{sec:unbias} and analyze the regret of DBGD-DSP in Section \ref{sec:regret}. The same procedure and conclusions can be applied to any DBGD-type algorithm of interest. 

\subsection{Problem Setup}\label{sec:setup}
The estimation of OL2R models can be formalized as a dueling bandit problem \cite{yue2009interactively}. In iteration $t$, an OL2R algorithm receives a query and associated candidate documents, which are represented as a set of $d$-dimensional query-document pair feature vectors $X_t = \{x_1, x_2, ..., x_s\}$. The algorithm takes two actions: first, it proposes two rankers, whose parameters are denoted as $w, w'$; second, it ranks the given documents with these two rankers accordingly. An oracle (i.e., user) compares (duels) the two rankers' results and provides feedback. In practice, an interleaving method \cite{radlinski2008does} is applied to merge the ranking lists of the two rankers and display the resulting ranked list to the user. User preference is inferred from the click feedback. Thus, the ranker that contributes more clicked documents is preferred. We denote $w \succ w'$ for the event that $w$ is preferred over $w'$. The comparison between two individual rankers is determined independently of other comparisons performed before with a probability $P\left(w \succ w' | X_t\right)$, such that $P\left(w \succ w' | X_t\right) = P_t\left(w \succ w' \right)  = f_t(w, w')$. %$f(w, w') \in \left[-\frac{1}{2}, \frac{1}{2}\right]$ 
$f_t(w, w')$ can be viewed as the distinguishability of the two rankers $w$ and $w'$ by an interleave comparison under query $X_t$.

We quantify the performance of an online learning algorithm using cumulative regret defined as follows:
\begin{equation}\label{eq:regret}
    R(T) = \sum_{t=1}^T f_t(w^*, w_t) + f_t(w^*, w'_t),
\end{equation}
where $w_t$ and $w'_t$ are rankers compared at time $t$, and $w^*$ is the best ranker in ground-truth. As a result, the distinguishability measure $f_t(w^*, w)$ indicates the loss of proposing a sub-optimal ranker $w$. We denote $f_t(w_t, w)$ as $f_t(w)$ for simplicity. The goal of an OL2R algorithm is to optimize its parameter towards $w^*$ according to loss $f_t(w)$.  A desired OL2R algorithm should have a sublinear regret in a finitie time horizon $T$, so that the one-step regret is quickly decreasing to zero over time.

%\subsubsection{Modeling Assumption}
In this work, we make the following assumptions similar to \cite{yue2009interactively}. We assume an unknown utility function $v_t(w)$ that quantifies the quality of a ranker $w$ over query $X_t$. 
The utility function $v_t$ is assumed to be  differentiable, strongly concave and $L_v$-Lipschitz, which means $\lvert v_t(x)-v_t(y) \rvert\leq L_v \lvert x-y\rvert$. 

A link function $\sigma$ describes the probabilistic comparison of utilities of two rankers as,
\begin{equation*}
    P_t\left(w \succ w'\right) = f_t(w, w') = \sigma\left(v_t(w)-v_t(w')\right).
\end{equation*}
The link function should be rotation-symmetric, which means $\sigma(x)=1-\sigma(-x)$. We assume the link function is $L_\sigma$-Lipschitz and second order $L_2$-Lipschitz. The link function behaves like a cumulative probability distribution function.
For example, a common choice of link function is the standard logistic function $\sigma(x) = \frac{1}{1+\exp(-x)}$ , which satisfies all the assumptions.  %Link functions often behave like cumulative probability distribution functions

\subsection{Document Space Projected Dueling Bandit Gradient Descent}\label{sec:alg}

\begin{algorithm}[t]
\caption{Document Space Projected Dueling Bandit Gradient Descent (DBGD-DSP)} \label{alg}
\begin{algorithmic}[1]
\State \textbf{Inputs: } $\delta, \alpha$
\State Initiate $w_1 = \textit{sample\_unit\_vector}()$
\For{ $t=1$ to $T$}	
\State Receive query $X_t = \{x_1, x_2, ..., x_s\}$
\State $u_t = \textit{sample\_unit\_vector}() $
\State $w_t' = w_t + \delta u_t$
\State Generate ranked lists $l(X_t, w_t), l(X_t, w_t')$

\State Set $L_t=\text{Interleave}\big(\{l(X_t, w_t), l(X_t, w_t')\}\big)$, and present $L_t$ to user
\State Receive click positions $C_t$ on $L_t$, and infer click credits ${\{c_t, c_t'\}}$
\If {  $c_t \geq c_t'$ } 
\State $w_{t+1} = w_t$
\Else
\State Based on $C_t$, infer user examined top $m_t$ documents in $L_t$.
\State Solve the orthogonal projection matrix $\bA_t$ for document space $S_t = span(\{x_{L_t, 1}, x_{L_t, 2}, ..., x_{L_t, m_t}\})$. 
\State Project $u_t$ onto $S_t$ by $g_t = \bA_t u_t$
% \State $w_{t+1} = w_t +\alpha \bA_t u_t$
\State $w_{t+1} = w_t +\alpha g_t$
%\State $w_{t+1}^0 = w_t^0 +\alpha \mathbf{Project}(u_t, S_t)$
\EndIf
\EndFor
\end{algorithmic}
\end{algorithm}

We describe our proposed Document Space Projected Dueling Bandit Gradient Descent (DBGD-DSP) in Algorithm \ref{alg}. We should note it fits all DBGD-type OL2R algorithm settings. At the beginning of iteration $t$, user initiates a query $X_t$. We denote $w_t$ as the parameter of the current ranker. DBGD-DSP first uniformly samples a vector $u_t$ from $d$ dimensional unit sphere $\bbS^{d-1}$ (i.e., $|u_t|_2=1$) as an exploratory direction, and proposes a candidate ranker $w_t' = w_t + \delta u_t$, where $\delta$ is the step size of exploration.
The algorithm then uses the two rankers ($w_t$ and $w_t'$) to generate ranking lists $l(X_t, w_t)$ and $l(X_t, w_t')$ accordingly, and combines them with an interleaving method, such as Team Draft Interleaving \cite{radlinski2008does} or Probabilistic Interleaving \cite{hofmann2011probabilistic}. The user examines the result list and provides implicit click feedback to indicate their relevance evaluation of the results. The interleaving method uses this implicit feedback to infer which ranker is preferred by the user. 
If the exploratory ranker is preferred (i.e., wins the duel), previous DBGD-style algorithms update the current ranker by $w_{t+1} = w_t + \alpha u_t$, where $\alpha$ is the learning rate; otherwise the current ranker stays intact. This gradient exploration strategy yields an unbiased estimate of the true gradient \cite{flaxman2005online}, in terms of expectation. %, i.e., $\nabla\hat f_t(w)=\bbE_{u_t \in \bbS^{d-1}}[f_t(w+\delta u)u]d/\delta$ where $\hat f_t(w)$ is a smoothed version of $f_t(w)$. 

However, since the exploratory gradient $u_t$ is required to be uniformly sampled from the entire $d$ dimensional unit sphere $\bbS^{d-1}$, the model update suffers from high variance in its gradient estimation, especially when $d$ is large, as in practice. Various improvements to this issue have been proposed in the past, but they still introduce other difficulties, such as variance and bias trade-off \cite{oosterhuis2017balancing,wang2018efficient,hofmann2013reusing}, and test sensitivity and efficiency \cite{schuth2014multileaved,zhao2016constructing}.
%To reducing the variance in gradient estimation, Schuth et al. \cite{schuth2014multileaved} proposed the Multileaving Gradient Descent algorithm to uniformly explore multiple directions at the same time via multileaved test. Zhao and King \cite{zhao2016constructing} proposed the Dual-Point Dueling Bandit Gradient Descent algorithm to explore two opposite directions in each iteration. By testing multiple exploratory directions, the variance of model estimation is reduced. However since in practice user only examines a finite number of documents, the sensitivity of interleaved test is limited, which prevents the algorithms from testing more rankers. Another line of research focus to constrain the sampling space of gradient exploration to reduce the variance. Oosterhuis et. al \cite{oosterhuis2017balancing} proposed to explore gradients in a subspace constructed by a set of preselected reference documents. Wang et al. \cite{wang2018efficient} proposed to use historical interactions to avoid repeatedly exploring less promising gradients. However, the variance of gradient exploration is reduced at a cost of introducing potential bias into gradient approximation, and therefore such algorithms could risk of converging to sub-optimal results. %\huazheng{Do we need to reduce this paragraph?} 

\begin{figure}[t]
\centering
\setlength\tabcolsep{1pt}
\begin{tabular}{c}
\hspace*{0cm}
\includegraphics[width=8cm]{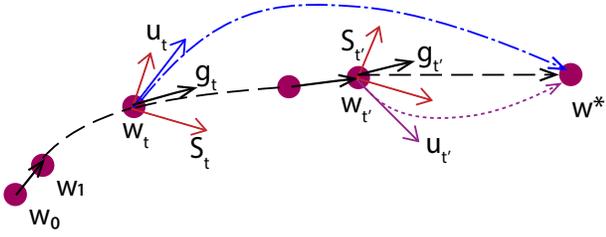} \\
\end{tabular}
\vspace{-2mm}
\caption{Illustration of model update for DBGD-DSP in a three dimensional space. Dashed lines represent the trajectory of DBGD following different update directions. $u_t$ is the selected direction by DBGD, which is in the 3-d space. Red bases present the document space $S_t$ on a 2-d plane. $u_t$ is projected onto $S_t$ to become $g_t$ for model update.} \label{Fig:procedure}
\vspace{-2mm}
\end{figure}

Unlike previous works that reduce the sampling space of gradient exploration before the interleaved test \cite{oosterhuis2017balancing,wang2018efficient,hofmann2013reusing}, we change the winning direction after the test. The key insight is that only the projected true gradient in the spanned space of \emph{examined} documents under query $X_t$ (denoted as document space $S_t$) can be revealed by an interleaved test. For example, as shown in Figure \ref{Fig:procedure}, a DBGD-style algorithm is comparing the current ranker $w_t$ and $ w_t' =w_t + \delta u_t$ with a uniformly sampled exploration direction $u_t$. The user examines top $m$ documents, e.g., $\{x_1, .. x_m\}$, of the interleaved ranking list (of course $m$ is unknown to the algorithm) and $w_t'$ wins the duel. The estimated gradient $u_t$ can therefore be separated into two components, one component $g_t$ that belongs to the document space $S_t = span\{x_1, .. x_m\}$ and the other component $u_t - g_t$ that is orthogonal to document space $S_t$. The orthogonal component $u_t - g_t$ does not affect the ranking among the examined documents, i.e. $(w_t+\delta u_t)^T x_i = (w_t+\delta g_t)^T x_i$, and thus does not contribute to the loss function and true gradient estimation. Intuitively, $u_t - g_t$ is not supported by the observed interleaved test, as anything sampled from the complement of $S_t$ cannot be verified by the examined documents. As a result, it is safe to exclude the direction $u_t - g_t$ from model update, which we later prove maintains the unbiasedness of the original DBGD-type gradient estimation, and reduces the variance. As illustrated in Figure \ref{Fig:procedure}, although $u_t$ will eventually lead to the same model estimation, as it is unbiased, this guarantee is only obtained in expectation. The variance could potentially be large: for example, the blue and purple updating traces slow down model convergence, when the number of observations is finite. 

As shown in line 14 to 16 of Algorithm \ref{alg}, we solve for the orthogonal projection matrix $\bA_t$ of document space $S_t$, and project the selected direction $u_t$ onto the document space $S_t$ after each interleaved test. 
%The projected gradient $g_t = \bA_t u_t$, which removes the component that is unsupported by user examination and orthogonal to document space, is used for model update. 
We leave the detailed design of constructing document space and solving projection matrix $\bA_t$ in Section \ref{sec:practice}. 
%Because users' result examination is affected by the ranked results, which are in turn determined by the proposed exploratory directions, restricting the exploration space introduces bias in subsequent interleaved test and model update. Instead, our projected gradient rectifies the gradient $u_t$, which explores the full parameter space, after interleaved test, thus  is an unbiased estimate of true gradient.
Before that, we first rigorously prove the projection maintains an unbiased estimate of true gradient in Section \ref{sec:unbias}. Since the document space is constructed only by the examined documents, the rank of document space is expected to be smaller than the entire parameter space. This directly leads to lower variance and faster model convergence. We show that our document space projection reduces the variance of gradient estimation from $d$ to $Rank(\bA_t)$ in Section \ref{sec:regret}, and then analyze its benefit for regret reduction from a low-variance gradient estimation. 

\subsection{Unbiasedness of Gradient Estimation}\label{sec:unbias}
We now prove that our document space projected gradient is an unbiased estimate of true gradient in the sense of expectation \cite{yue2009interactively}. We define $Z_t(w)$ as the event of $w$ winning the duel with $w_t$,
\[Z_t(w) = \left\{
	\begin{array}{ll}
		1 & \mbox{w.p. }~ 1-P_t(w_t\succ w)  \\
		0 & \mbox{w.p. }~ P_t(w_t\succ w)
	\end{array}
\right. \]
Then the gradient used for model update in DBGD-DSP (as described in Algorithm \ref{alg}) can be described as, 
\begin{equation}
    h_t = -Z_t(w_t + \delta u_t)g_t.
\end{equation}
Note that by adding a negative sign we view our model update as online gradient descent $w_{t+1} = w_t - \alpha g_t$. 

We now show in the following theorem that this is an unbiased gradient estimation of true gradient. By defining a smoothed version of $f_t$ as $\hat f_t(w) = \bbE_{u \in \bbB}[f_t(w+\delta u)]$, we have:
\begin{theorem}\label{theorem:unbias} 
The projected gradient $g_t$ in DBGD-DSP  is an unbiased estimate of true gradient, i.e.,
\begin{equation}
    \bbE[h_t] = \frac{\delta}{d}\nabla \hat f_t(w)
\end{equation}
over random unit vector $u_t$.
\end{theorem}

\begin{proof}
Based on the Lemma 1 of \cite{yue2009interactively}, we have
\begin{align*}
\bbE\left[h_t\right] &= \bbE\left[-Z_t(w_t + \delta u_t)\bA_t u_t\right] =\bbE_{u_t\in \bbS^{d-1}}\left[f_t(w+\delta \bA_t u_t)u_t\right]
%\bbE_{u}\left[f_t(w+\delta u)u\right] = \frac{\delta}{d}\nabla \hat f_t(w)
\end{align*}
Define $F_t(w) = f_t(\bA_t w)$, we have
\begin{align*}
\bbE[h_t] &= \bbE_{u_t\in \bbS^{d-1}}\left[f_t(w_t+\delta \bA_t u_t)u_t\right]\\
&= \bbE_{u_t\in \bbS^{d-1}}[F_t(\bA_t^{-1}w_t+\delta u_t)u_t]\\
&= \frac{\delta}{d} \nabla \bbE_{u_t\in \bbB^d}[F_t(\bA_t^{-1}w_t+\delta u_t)u_t]\\
&= \frac{\delta}{d} \nabla \hat F_t(\bA_t^{-1}w_t)\\
&= \frac{\delta}{d} \bA_t \nabla \hat f_t(w_t)\\
&= \frac{\delta}{d} \nabla \hat f_t(w_t)\end{align*}
where the third equality is based on Stokes' Theorem. The last equality holds because gradient $\nabla \hat f_t(w_t)$ belongs to document space $S_t$, and thus projecting it by $\bA_t$ maps back to itself.
\end{proof}
%Theorem \ref{theorem:unbias} suggests that DBGD-DSP is performing gradient descent over $\hat f_t(w)$.
The guarantee of unbiased gradient estimation is a major advantage of our proposed document space gradient projection method, compared with previous attempts to reduce the gradient exploration space, such as Oosterhuis et. al \cite{oosterhuis2017balancing} and Wang et al. \cite{wang2018efficient}. Our method enjoys reduced variance of gradient estimate (which will be proved next), without the risk of converging towards a sub-optimal solution. We should note that the above is independent from the mechanism of how the proposal directions are generated, as shown in the first four steps of proof above. As a result, if the input direction to our projection procedure is unbiased, the resulting update direction is also unbiased. This enables our solution's generalization to other types of DBGD algorithms. 

%the convergence speed 

\subsection{Regret Analysis of DBGD-DSP} \label{sec:regret}
We now analyze the regret of our proposed DBGD-DSP algorithm, starting with its variance of gradient update.

\begin{lemma}\label{lemma:variance} 
The variance of gradient update in DBGD-DSP is bounded by 
\begin{equation*}
%\label{eq_variance}
\bbE[|h_t|^2] = \bbE_{u_t \in \bbS^{d-1}}\left[|-Z_t(w_t + \delta u_t)\bA_t u_t|^2\right] \leq \frac{Rank(\bA_t)}{d}.  
\end{equation*}
\end{lemma}
\begin{proof}
\begin{align*}
\bbE[|h_t|^2] &= \bbE_{u_t}\left[|-Z_t(w_t + \delta u_t)\bA_t u_t|^2\right]\\
&\leq \bbE_{u_t}\left[|\bA_t u_t|^2\right]\\
&=\bbE_{u_t}\left[(\bA_t u_t)^\top (\bA_t u_t) \right]\\
& = \Tr\left(\bbE_{u_t}\left[\bA_t u_t u_t^\top \bA_t^\top \right]\right) \text{ //apply the trace trick}\\
& = \Tr\left(\bA_t\bbE_{u_t}\left[ u_t u_t^\top  \right]\bA_t^\top\right)\\
&= \Tr\left(\bA_t \frac{1}{d}I \bA_t^\top\right)\\
&= \frac{1}{d}\Tr\left(\bA_t \bA_t^\top\right)\\
&= \frac{1}{d}\Tr\left(\bA_t\right) \text{ //a projection matrix is idempotent} \\
&= \frac{\Rank(\bA_t)}{d} 
\end{align*}
where $\Tr(\cdot)$ denotes the matrix trace operation. The sixth equality holds because $u_t$ is uniformly sampled from a unit sphere, and its covariance matrix $\bbE_{u_t}\left[ u_t u_t^\top  \right]$ is $\frac{1}{d}I$.  Since $\bA_t$ is an orthogonal projection matrix, the eighth equality holds for $\bA_t \bA_t^\top = \bA_t $.  
\end{proof}
\begin{remark}
The variance of gradient update in DBGD \cite{yue2009interactively} is bounded by $\bbE_{u_t}\left[|-Z_t(w_t + \delta u_t) u_t|^2\right]\leq 1$.
\end{remark}
Comparing the variance of gradient update in DBGD-DSP with DBGD, our method reduces the variance from $1$ to $\frac{\Rank(\bA_t)}{d}$. Since the dimension of projection matrix $\bA_t$ is $d$-by-$d$, we have $\Rank(\bA_t)\leq d$, which guarantees the reduction of variance in DBGD-DSP comparing to that in DBGD. The rank of $\bA_t$ is also bounded by the number of \emph{examined} documents $m_t$, since document space $S_t$ is constructed by these $m_t$ examined documents. In practice, users would only examine a handful of documents \cite{joachims2017accurately,craswell2008experimental}, while the ranking feature dimension is expected to be much larger. We argue that $m_t \ll d$, such that our document space projection achieves considerable variance reduction. 

The significance of this variance reduction can be intuitively understood from Figure \ref{Fig:procedure}: though different traces of model update would eventually lead to the same converged model, if one has a sufficiently large amount of interactions with users, the one with lower variance would always require less observations. A faster converging algorithm leads to user satisfaction earlier. Next, we verify this benefit by proving the reduction of regret introduced by the reduced variance in gradient estimation.  
\begin{theorem}\label{theorem:regret} 

By setting
\begin{equation*}
    m = \max_t m_t, \delta= \frac{\sqrt{2Rm}}{\sqrt{13L}T^{1/4}}, \alpha=\frac{Rm}{\sqrt{T}\delta},
\end{equation*}
the expected regret of DBGD-DSP as defined in Eq \eqref{eq:regret} is upper bounded by, 
\begin{equation}
    \bbE[Reg] \leq 2\lambda_T T^{3/4}\sqrt{26RmL},
\end{equation}
where
\begin{equation*}
\lambda_T = \frac{L_\sigma\sqrt{13L}T^{1/4}}{L_\sigma\sqrt{13L}T^{1/4} - L_vL_2\sqrt{2Rm}}
\end{equation*}
\end{theorem}
The proof is obtained by extending Theorem 2 in \cite{yue2009interactively}. %We provide the proof sketch below to show the key difference in regret of DSP-DBGD.
We omit the details due to space limit, and emphasize that the key difference is introduced by replacing variance of gradient estimation from $\bbE_{u_t}\left[|-Z_t(w_t + \delta u_t) u_t|^2\right]$ to $\bbE_{u_t}\left[|-Z_t(w_t + \delta u_t) \bA_t u_t|^2\right]$. Since the variance of gradient estimation is reduced from 1 to $\frac{\Rank(\bA_t)}{d}$, the regret of DBGD can be reduced from $O(\sqrt{d}T^{3/4})$ to $O(\sqrt{m}T^{3/4})$, where $m$ is the maximum number of documents included in a document space under a single query. Again, as the number of included ranking features is oftentimes much larger than the number of documents a user would examine under a single query, the reduction of regret is considerable. Moreover, as the reduction of variance from our project-based method is independent from the way about how the proposal directions are generated, our method can be generally applied to most existing DBGD-type OL2R algorithms to improve their learning convergence.

\subsection{Practical Treatments of Document Space Projection} \label{sec:practice}
Now we discuss several practical treatments of our proposed Document Space Projection method, including the construction of document space and orthogonal projection matrix. 

In our theoretical analysis, we have assumed the knowledge of users' examined documents and corresponding projection matrix. However, in practice, a user's result examination is unobserved. A rich body of research has been developed to perform statistical inference of it, collectively known as click modeling \cite{chapelle2009dynamic,craswell2008experimental}. Any of these existing click models can be plugged into our solution framework, i.e., line 13 of Algorithm \ref{alg}. In this work, we simply follow \cite{joachims2017accurately} to infer user examination by the last clicked position: given the click position list $C_t$, we use the last clicked position $c_{l,t}$ to approximate the last examined position $M_t$ by setting $M_t = c_{l,t}+k$, where $k$ is a hyper-parameter. Based on sequential examination hypothesis of click modeling, every document before the last clicked position is examined, and we use $k$ to approximate the number of positions following the last clicked position that were still examined. We leave more comprehensive study of click modeling in our solution as future work. 

The above treatment provides a reasonable inference of examined documents. However, it requires a careful choice of $k$ for each query (preferably). If $k$ is set too large, the variance of gradient estimate will increase (as proved in Lemma \ref{lemma:variance}). If $k$ is too small, the document space may not include all examined documents, and it is at risk of introducing bias in gradient projection. To avoid bias in constructing the document space, we also consider adding historically examined documents to the current query's document space. Specifically, we add $r$ recently examined documents to the current document space $S_t$ to compensate the potentially overlooked examined documents in the current query.

In line 14 of Algorithm \ref{alg}, we solve the orthogonal projection matrix $\bA_t$ of document space $S_t$. $\bA_t$ could be computed by several methods. Denote $D_t$ as a $d$-by-$m_t$ matrix where each column is the feature vector for an examined document. One can use QR decomposition or Singular Value Decomposition (SVD) to solve for its orthonormal basis $V_t$, and the projection matrix can then be constructed by $\bA_t = V_t V_t^T$. In our experiments, we chose SVD for constructing the basis of document space, because of its widely available and efficient large-scale implementations. But the choice for the construction of this project matrix does not affect the convergence nor unbiasedness of our proposed solution. 

\section{Experiments}\label{sec:exp}

To demonstrate our proposed Document Space Projection method's empirical efficacy, we compare the performance of several best-performing DBGD-type OL2R algorithms on five public learning to rank datasets, with and without our document space projection method applied. %An interesting feature of the basic DSGD algorithm is that it leaves room for choosing different sampling methods within the document space. We tested several different sampling methods empirically, and these different methods provide interesting and varied results.

\subsection{Experiment Setup}
\label{exp:setup}
\noindent\textbf{$\bullet$ Datasets.}
We tested our algorithms and the baselines on five benchmark datasets: including MQ2007, MQ2008, NP2003 \cite{liu2007letor}, MSLR-WEB10K \cite{qin2013introducing}, and the Yahoo! Learning to Rank Challenge dataset \cite{chapelle2011yahoo}. In each of the five datasets, each query-document pair is encoded as a vector of ranking features. These features include PageRank, TF.IDF, Okapi-BM25, URL length, language model score, and many more varied by dataset.

The MQ2007 and MQ2008 datasets are collected from the 2007 and 2008 Million Query track at TREC \cite{voorhees2005trec}. MQ2007 contains about 1700 queries, and MQ2008 contains about 800 queries, which represent a mix of informational and navigational search intents. They both have 46-dimensional feature vectors to represent query-document pairs, and the document relevance are labeled in three grades: 0 (not relevant), 1 (relevant), and 2 (most relevant).

The NP2003 dataset also comes from the TREC Web track, consisting of queries crawled from the .gov domain. It is comprised of about 150 navigational-focused queries, with over 1000 document relevance assessments per query. It uses 64 ranking features, and the document relevance labels are binary (0 and 1 only). 

The MSLR-WEB10K dataset was released by Microsoft in 2010, and consists of 10,000 queries with relevance assessments coming from a labeling set from the Microsoft Bing search engine. It has 136 ranking features, and the relevance judgments range from 0 (not relevant) to 4 (most relevant).

The Yahoo! Learning to Rank Challenge dataset was also released in 2010, as an effort on  part of Yahoo! to promote the dataset as well as research into better learning to rank algorithms. The dataset contains about 36,000 queries, 883,000 assessed documents, and 700 ranking features. Again, the relevance judgments range from 0 (not relevant) to 4 (most relevant)
 
This diversity in the structure of the datasets that we chose to test on helps us to evaluate our algorithms more holistically. While small, the MQ2007 and MQ2008 sets have been around for a long time and have a good mix of query types. NP2003 gives us insight into how the algorithms perform on navigational search intents specifically, which are markedly different in nature from informational search intents. MSLR-WEB10K and the Yahoo! dataset are large-scale datasets used by actual commercial search engines, which give us a better understanding of how the algorithms perform in practice. Since each dataset was split into training, testing, and validation subsets, we used the training sets for online experiments to measure cumulative performance, and used the testing sets for evaluating offline performance.

\noindent\textbf{$\bullet$ Simulated User Interactions.}
Based on an online learning to rank framework proposed in \cite{Oosterhuis_2018}, we use the standard setup to simulate user interactions. Within this framework, we used the Cascade Click Model to simulate user click behavior. This model assumes that a user interacts with a set of search results by linearly scanning the list from top and making a decision for each document as to whether or not to click. In the model, the probability of a click for a given document is conditioned on the relevance label of that document, as a user is expected to be more likely to click on relevant documents. After evaluating each document, the user must decide whether or not to continue perusing the list. This decision's probability distribution is again conditioned on the relevance of the examined document, as a user is more likely to stop looking through the results if he/she has already satisfied their information need. These aforementioned probabilities can be altered to simulate different types of users and interactions.

\begin{figure*}[ht]
\centering
\setlength\tabcolsep{4pt}
\vspace{-1mm}
\begin{tabular}{ccc}
\hspace*{0cm}
\includegraphics[width=5.7cm]{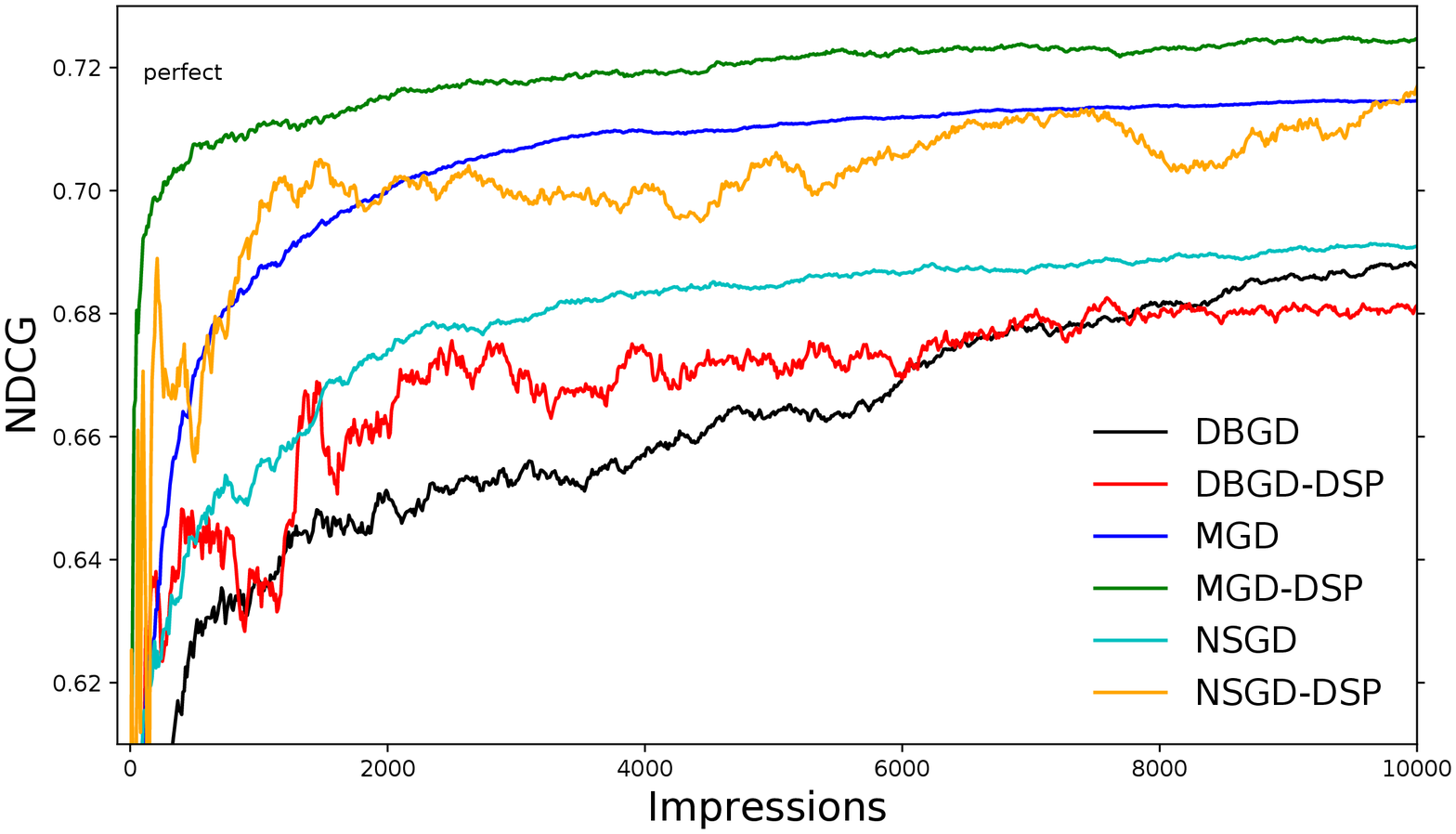} & 
\includegraphics[width=5.7cm]{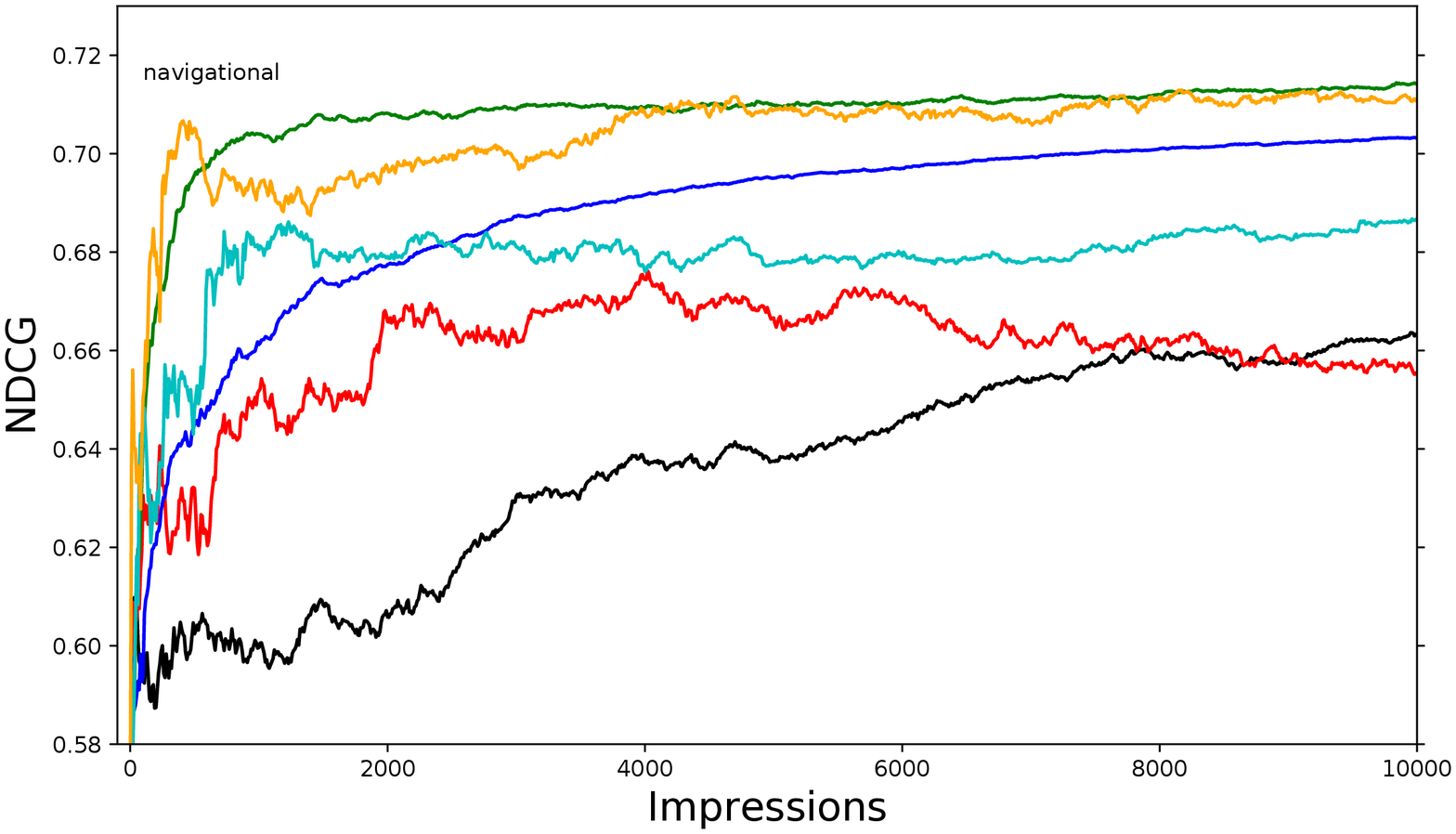} & 
\includegraphics[width=5.7cm]{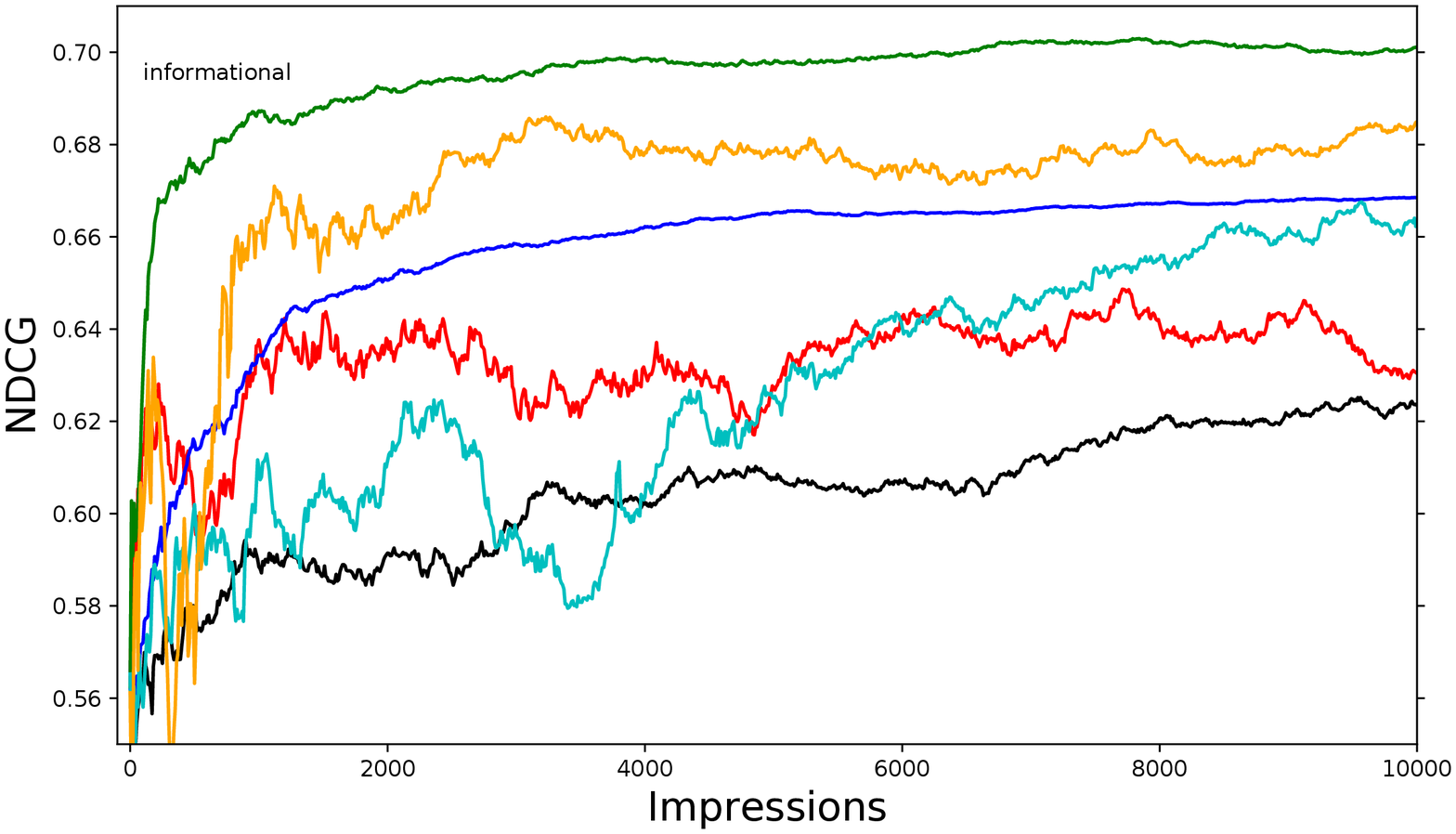} \\
(a) Perfect & (b) Navigational & (c) Informational \\
\end{tabular}
\vspace{-2mm}
\caption{Offline NDCG@10 on Yahoo! dataset.} \label{Fig:offline}
\vspace{-2mm}
\end{figure*}

As illustrated in Table \ref{tab:clickmodel}, we use three different click model probability configurations to represent three different types of users. First, we have the \textit{perfect} user, who clicks on all relevant documents and does not stop browsing until they have visited all of the documents. This type of users contribute the least noise, as they make no mistakes and the feedback is entirely accurate. Second, we have the \textit{navigational} user, who is very likely to click on the first highly relevant document that he/she sees and stops there. Third, we have the \textit{informational} user, who, in his/her search for information, sometimes clicks on irrelevant documents, and as such contributes a significant amount of noise in click feedback. % As a standard setting set forth in \cite{schuth2016multileave,zhao2016constructing}, the length of the results list that these cascade click models evaluate is 10 for our experiments.

\begin{table}[t]
    %\vspace{-3mm}
    \centering
    \caption{Configurations of simulation click models.}\label{tab:clickmodel}
    \vspace{-1mm}
    \begin{tabular}{ccccccccccc}
    \hline
         & \multicolumn{5}{c}{Click Probability} & \multicolumn{5}{c}{Stop Probability} \\
        R & 0 & 1 & 2 &3&4 & 0 & 1 & 2 & 3 & 4\\
         \hline
        Per & 0.0 & 0.2 & 0.4 & 0.8 & 1.0 & 0.0 & 0.0 & 0.0 & 0.0 & 0.0\\
        Nav & 0.05 & 0.3 & 0.5 & 0.7 & 0.95 & 0.2 & 0.3 & 0.5 & 0.7 & 0.9\\
        Inf & 0.4 & 0.6 & 0.7 & 0.8 & 0.9 & 0.1 & 0.2 & 0.3 & 0.4 & 0.5\\
        \hline
    \end{tabular}
    \vspace{-3mm}
\end{table}

\noindent\textbf{$\bullet$ Evaluation Metrics.}
As set forth in \cite{schuth2013lerot}, cumulative (online) Normalized Discounted Cumulative Gain (NDCG) and offline NDCG are commonly used metrics for evaluating OL2R algorithms. Cumulative NDCG is calculated by summing NDCG scores from successive iterations with a discount factor $\gamma$ set to 0.995. We assess our model's estimation convergence via cosine similarity between the current weight vector and a reference weight vector (considered to be the optimal vector) as estimated by an offline learning-to-rank algorithm trained with the complete true relevance judgment labels. Due to its superior empirical performance, we used LambdaRank \cite{burges2010ranknet} with no hidden layer in our experiments to estimate this reference weight vector. In each experiment, the number of iterations $T$ was set to 10,000, and the current query $X_t$ was randomly sampled from the dataset in each iteration. We execute all the experiments 15 times with different random seeds, and report and compare the average performance in all experiments.  

\noindent\textbf{$\bullet$ Evaluation Questions.}
To better understand the advantages of our proposed algorithms, we aim to answer the following evaluation questions through the course of our experiments.
\begin{itemize}
\item[Q1:] Can our proposed Document Space Projection method consistently improve the performance of best-performing DBGD-type OL2R algorithms?
\item[Q2:] Do gradients rectified by our document space projection explore the gradient space more efficiently?
\item[Q3:] How do different hyper-parameter settings alter the performance of our document space projection?
\end{itemize}

\noindent\textbf{$\bullet$ Baseline Algorithms.}
We choose the following three best-performing DBGD-type OL2R algorithms as our baselines for comparison:
\begin{itemize}
    \item[-] \textbf{DBGD}~\cite{yue2009interactively}: A single direction uniformly sampled from the whole parameter space is explored. %Team Draft is used to interleave the results of the two rankers for comparison.
    \item[-] \textbf{MGD}~ \cite{schuth2016multileave}: Multiple directions are explored in one iteration to reduce the gradient estimation variance. Multileaving is used to compare multiple rankers. The model updates towards the mean of all rankers that beat the current model.%If there is a tie, the model updates towards the mean of all winners.
    \item[-] \textbf{NSGD}\cite{wang2018efficient}: Multiple directions are sampled from the null space of previously poorly performing gradients. Ties are broken by evaluating the tied candidate rankers on a recent set of difficult queries.
\end{itemize}
We apply our proposed Document Space Projection to the baseline algorithms, and compare them with DBGD-DSP, MGD-DSP and NSGD-DSP, respectively.

\begin{table*}[ht] 
    \centering
    %\vspace{-1mm}
    \caption{Online NDCG@10, standard deviation and relative improvement of document space projection of each algorithm after 10,000 queries. } %Statistically significant improvements over MGD baseline are indicated by $\blacktriangle$ (p<0.05).} 
    \label{Tab:online}
    %\vspace{-2mm}
\begin{tabular*}{\textwidth}{@{\extracolsep{\fill} } l  l l l l l  l}
\toprule
 { \small Click Model} & { \small Algorithm}  & { \small \textbf{MQ2007}}  & { \small \textbf{MQ2008}}  & { \small \textbf{MSLR-WEB10K}}  & { \small \textbf{NP2003}}  & { \small \textbf{Yahoo}} \\
\midrule 
%& \multicolumn{5}{|c|}{\textit{perfect}} \\
%\midrule
\multirow{6}{*}{Perfect} &
DBGD& 679.3 {\tiny (21.6)} & 847.1 {\tiny (38.4)} & 532.2 {\tiny (15.3)} & 1130.2 {\tiny (43.3)} & 1165.5 {\tiny (22.6)} \\
&DBGD-DSP & 689.1 {\tiny (19.5)}(+1.44\%) & 858.0 {\tiny (39.2)}(+1.29\%) & 553.6 {\tiny (13.1)}(+4.02\%) & 1198.8 {\tiny (40.0)} (+6.07\%)& 1198.8 {\tiny (33.5)}(+2.86\%) \\
\cline{2-7}
&MGD& 689.1 {\tiny (14.6)} & 859.4 {\tiny (38.1)} & 558.3 {\tiny (7.0)} & 1192.9 {\tiny (44.6)} & 1201.9 {\tiny (16.3)} \\
&MGD-DSP & \bf 757.3 {\tiny (16.2)}(+9.90\%) & \bf 919.5 {\tiny (42.2)}(+6.99\%) & 626.4 {\tiny (9.6)}(+12.20\%) & 1335.3 {\tiny (39.1)}(+11.94\%) & \bf 1309.4 {\tiny (10.6)} (+8.94\%)\\
\cline{2-7}
&NSGD & 684.4 {\tiny (20.5)} & 867.5 {\tiny (40.3)} & 589.5 {\tiny (14.2)} & 1274.9 {\tiny (47.4)} & 1162.3 {\tiny (12.9)} \\
&NSGD-DSP & 732.5 {\tiny (20.0)}(+7.03\%) & 904.3 {\tiny (38.0)}(+4.24\%) & \bf 635.6 {\tiny (12.8)}(+7.82\%) & \bf 1368.5 {\tiny (41.1)}(+7.34\%) & 1270.1 {\tiny (2.5)}(+9.27\%) \\
\midrule
\multirow{6}{*}{Navigational} &DBGD& 646.1 {\tiny (23.4)} & 817.9 {\tiny (45.5)} & 517.5 {\tiny (20.9)} & 1062.3 {\tiny (55.4)} & 1133.3 {\tiny (40.8)} \\
&DBGD-DSP & 664.9 {\tiny (26.9)}(+2.91\%) & 830.3 {\tiny (44.1)}(+1.52\%) & 543.1 {\tiny (14.8)}(+4.95\%) & 1140.1 {\tiny (52.5)}(+7.32\%) & 1199.4 {\tiny (34.6)}(+5.83\%)\\
\cline{2-7}
&MGD& 632.7 {\tiny (15.5)} & 827.5 {\tiny (35.5)} & 538.2 {\tiny (7.2)} & 1115.4 {\tiny (44.6)} & 1171.3 {\tiny (20.4)} \\
&MGD-DSP & \bf 694.5 {\tiny (15.7)}(+9.77\%) & \bf 882.3 {\tiny (40.0)}(+6.62\%) & 586.9 {\tiny (9.5)}(+9.05\%) & \bf 1300.9 {\tiny (39.6)}(+16.63\%) & \bf 1290.2 {\tiny (15.3)} (+10.15\%)\\
\cline{2-7}
&NSGD & 660.1 {\tiny (24.5)} & 849.1 {\tiny (36.6)} & 562.1 {\tiny (18.8)} & 1211.1 {\tiny (66.5)} & 1186.2 {\tiny (16.8)} \\
&NSGD-DSP & 724.6 {\tiny (24.5)}(+9.77\%) & 895.8 {\tiny (34.2)}(+5.50\%) & \bf 608.3 {\tiny (12.1)} (+8.22\%)& 1296.2 {\tiny (24.3)} (+7.03\%) & 1283.4 {\tiny (7.2)}(+8.19\%) \\
\midrule
\multirow{6}{*}{Informational} &
DBGD& 583.4 {\tiny (46.0)} & 763.9 {\tiny (55.1)} & 472.4 {\tiny (34.6)} & 849.8 {\tiny (144.5)} & 1107.3 {\tiny (46.6)} \\
&DBGD-DSP & 620.1 {\tiny (40.8)}(+6.29\%) & 782.4 {\tiny (51.8)} (+2.42\%)& 522.1 {\tiny (18.6)} (+10.52\%)& 992.5 {\tiny (81.1)}(+16.79\%) & 1158.5 {\tiny (22.0)}(+4.62\%) \\
\cline{2-7}
&MGD& 621.2 {\tiny (18.2)} & 817.5 {\tiny (45.3)} & 538.3 {\tiny (10.8)} & 1107.9 {\tiny (46.2)} & 1146.6 {\tiny (37.5)} \\
&MGD-DSP & \bf 671.4 {\tiny (18.9)}(+8.08\%) & \bf 865.9 {\tiny (37.7)}(+5.92\%) & 580.5 {\tiny (10.4)}(+7.84\%) & \bf 1274.5 {\tiny (42.9)}(+15.04\%) & \bf 1268.1 {\tiny (16.4)}(+10.60\%) \\
\cline{2-7}
&NSGD & 629.7 {\tiny (25.3)} & 814.9 {\tiny (37.1)} & 532.9 {\tiny (15.2)} & 1123.5 {\tiny (59.8)} & 1110.5 {\tiny (10.9)} \\
&NSGD-DSP & 703.6 {\tiny (29.2)}(+11.74\%) & 871.3 {\tiny (48.3)}(+6.92\%) & \bf 597.9 {\tiny (14.1)}(+12.20\%) & 1222.8 {\tiny (43.8)}(+9.03\%) & 1204.7 {\tiny (9.6)}(+8.48\%) \\
\bottomrule
\end{tabular*}
%\vspace{-2mm}
\end{table*}

\begin{table*}[ht] 
    \centering
    %\vspace{-1mm}
    \caption{Offline NDCG@10, standard deviation and relative improvement of document space projection of each algorithm after 10,000 queries. }% under each of the three click models. Statistically significant improvements over MGD baseline are indicated by $\blacktriangle$ (p<0.05).}
    %\vspace{-2mm}
    \label{Tab:offline}
\begin{tabular*}{\textwidth}{@{\extracolsep{\fill} } l  l l l l l  l}
\toprule
 { \small Click Model} & { \small Algorithm}  & { \small \textbf{MQ2007}}  & { \small \textbf{MQ2008}}  & { \small \textbf{MSLR-WEB10K}}  & { \small \textbf{NP2003}}  & { \small \textbf{Yahoo}} \\
\midrule 
%& \multicolumn{5}{|c|}{\textit{perfect}} \\
%\midrule
\multirow{6}{*}{Perfect} & DBGD& 0.484 {\tiny (0.023)} & 0.683 {\tiny (0.023)} & 0.331 {\tiny (0.009)} & 0.737 {\tiny (0.056)} & 0.688 {\tiny (0.011)} \\
&DBGD-DSP & 0.480 {\tiny (0.020)} (-0.83\%)  & 0.685 {\tiny (0.024)} (+0.29\%)  & 0.333 {\tiny (0.011)} (+0.6\%)  & 0.738 {\tiny (0.059)} (+0.14\%)  & 0.681 {\tiny (0.013)} (-1.02\%)  \\
\cline{2-7}
& MGD & 0.495 {\tiny (0.022)} & 0.691 {\tiny (0.020)} & 0.334 {\tiny (0.003)} & 0.746 {\tiny (0.048)} & 0.715 {\tiny (0.002)} \\
& MGD-DSP & \bf 0.501 {\tiny (0.021)}(+1.21\%)  & \bf 0.695 {\tiny (0.022)}(+0.58\%)  & \bf 0.409 {\tiny (0.006)}(+22.46\%)  & 0.748 {\tiny (0.055)}(+0.27\%)  & \bf 0.725 {\tiny (0.003)}(+1.40\%) \\
\cline{2-7}
& NSGD & 0.488 {\tiny (0.019)} & 0.689 {\tiny (0.024)} & 0.397 {\tiny (0.012)}  & 0.743 {\tiny (0.050)} & 0.691 {\tiny (0.005)} \\
& NSGD-DSP & 0.491 {\tiny (0.022)}(+0.61\%)  & 0.691 {\tiny (0.025)}(+0.29\%)  & \bf 0.398 {\tiny (0.008)} (+0.25\%) & \bf 0.750 {\tiny (0.042)} (+0.94\%) & 0.717 {\tiny (0.004)}(+3.76\%) \\
\midrule
% & \multicolumn{5}{|c|}{\textit{navigational}} \\
% \midrule
\multirow{6}{*}{Navigational}& DBGD& 0.463 {\tiny (0.028)} & 0.667 {\tiny (0.021)} & 0.320 {\tiny (0.012)} & 0.728 {\tiny (0.054)} & 0.663 {\tiny (0.020)} \\
&DBGD-DSP & 0.465 {\tiny (0.024)}(+0.43\%)  & 0.668 {\tiny (0.023)}(+0.15\%)  & 0.327 {\tiny (0.011)}(+2.19\%)  & 0.734 {\tiny (0.052)}(+0.82\%)  & 0.656 {\tiny (0.013)}(-1.06\%)  \\
\cline{2-7}
&MGD& 0.426 {\tiny (0.019)} & 0.664 {\tiny (0.016)} & 0.321 {\tiny (0.003)} & 0.740 {\tiny (0.048)} & 0.703 {\tiny (0.010)} \\

&MGD-DSP & 0.467 {\tiny (0.021)}(+9.62\%)  & \bf 0.684 {\tiny (0.017)}(+3.01\%)  & 0.331 {\tiny (0.005)}(+3.12\%)  & 0.744 {\tiny (0.053)}(+0.54\%)  & \bf 0.714 {\tiny (0.006)}(+1.56\%)  \\
\cline{2-7}
&NSGD & 0.473 {\tiny (0.022)} & 0.676 {\tiny (0.024)} & \bf 0.389 {\tiny (0.013)} & 0.732 {\tiny (0.053)} & 0.686 {\tiny (0.008)} \\
&NSGD-DSP & \bf 0.478 {\tiny (0.020)}(+1.06\%)  & \bf 0.683 {\tiny (0.026)}(+1.04\%)  & 0.376 {\tiny (0.014)}(-3.34\%) & \bf 0.788 {\tiny (0.006)}(+7.65\%)  & \bf 0.711 {\tiny (0.001)}(+3.64\%) \\
\midrule
% & \multicolumn{5}{|c|}{\textit{informational}} \\
% \midrule
\multirow{6}{*}{Informational} & DBGD& 0.410 {\tiny (0.034)} & 0.641 {\tiny (0.031)} & 0.294 {\tiny (0.022)} & 0.699 {\tiny (0.063)} & 0.623 {\tiny (0.037)} \\
&DBGD-DSP & 0.427 {\tiny (0.027)}(+4.15\%)  & 0.632 {\tiny (0.031)}(-1.4\%)  & 0.309 {\tiny (0.011)}(+32.65\%)  & 0.692 {\tiny (0.062)}(-1.00\%)  & 0.63 {\tiny (0.030)}(1.12\%)  \\
\cline{2-7}
&MGD& 0.406 {\tiny (0.020)} & 0.651 {\tiny (0.020)} & 0.317 {\tiny (0.003)} & 0.726 {\tiny (0.050)} & 0.668 {\tiny (0.044)} \\
&MGD-DSP & 0.444 {\tiny (0.025)}(+0.44\%)  & 0.669 {\tiny (0.018)}(+0.67\%) & 0.325 {\tiny (0.004)}(+0.33\%) & 0.738 {\tiny (0.054)}(+0.74\%) & \bf 0.701 {\tiny (0.005)}(+4.94\%)  \\
\cline{2-7}
&NSGD & \bf 0.469 {\tiny (0.018)} & \bf 0.674 {\tiny (0.023)} & \bf 0.360 {\tiny (0.013)} & 0.733 {\tiny (0.056)} & 0.663 {\tiny (0.015)} \\

&NSGD-DSP & 0.466 {\tiny (0.019)}(-0.64\%)  & 0.668 {\tiny (0.026)}(-0.89\%)  & 0.340 {\tiny (0.018)}(-5.56\%) & \bf 0.789 {\tiny (0.013)}(+7.64\%)  & 0.685 {\tiny (0.004)}(+3.32\%) \\
\bottomrule
\end{tabular*}
\vspace{-2mm}
\end{table*}

\begin{figure*}[ht]
\centering
\vspace{-5mm}
\setlength\tabcolsep{3pt}
\begin{tabular}{ >{\centering\arraybackslash}m{4.3cm} >{\centering\arraybackslash}m{4.3cm} >{\centering\arraybackslash}m{4.3cm}
>{\centering\arraybackslash}m{4.3cm}}

\hspace*{-3mm}
\includegraphics[width=4.5cm]{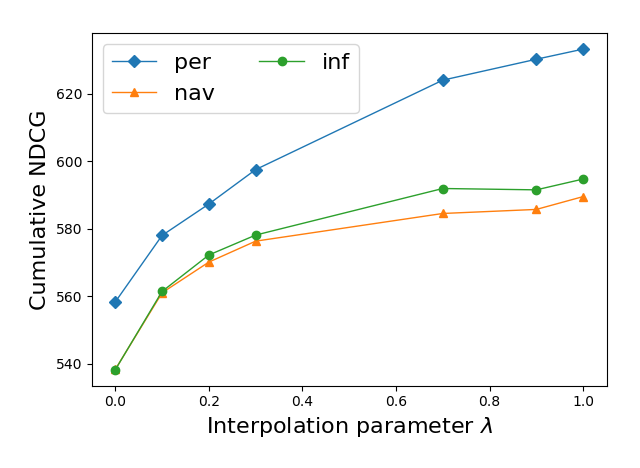} &
\includegraphics[width=4.5cm]{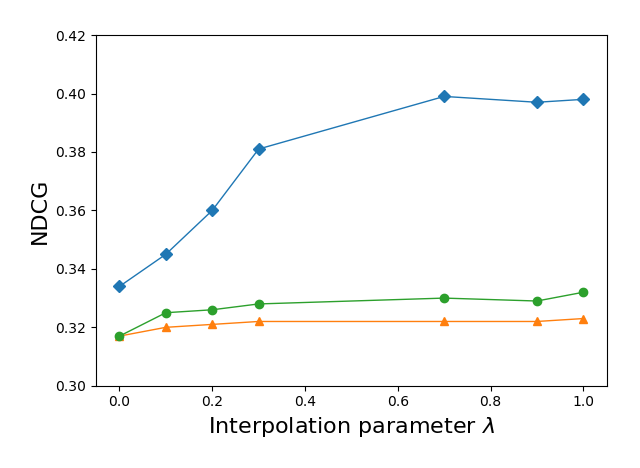}&
\includegraphics[width=4.7cm]{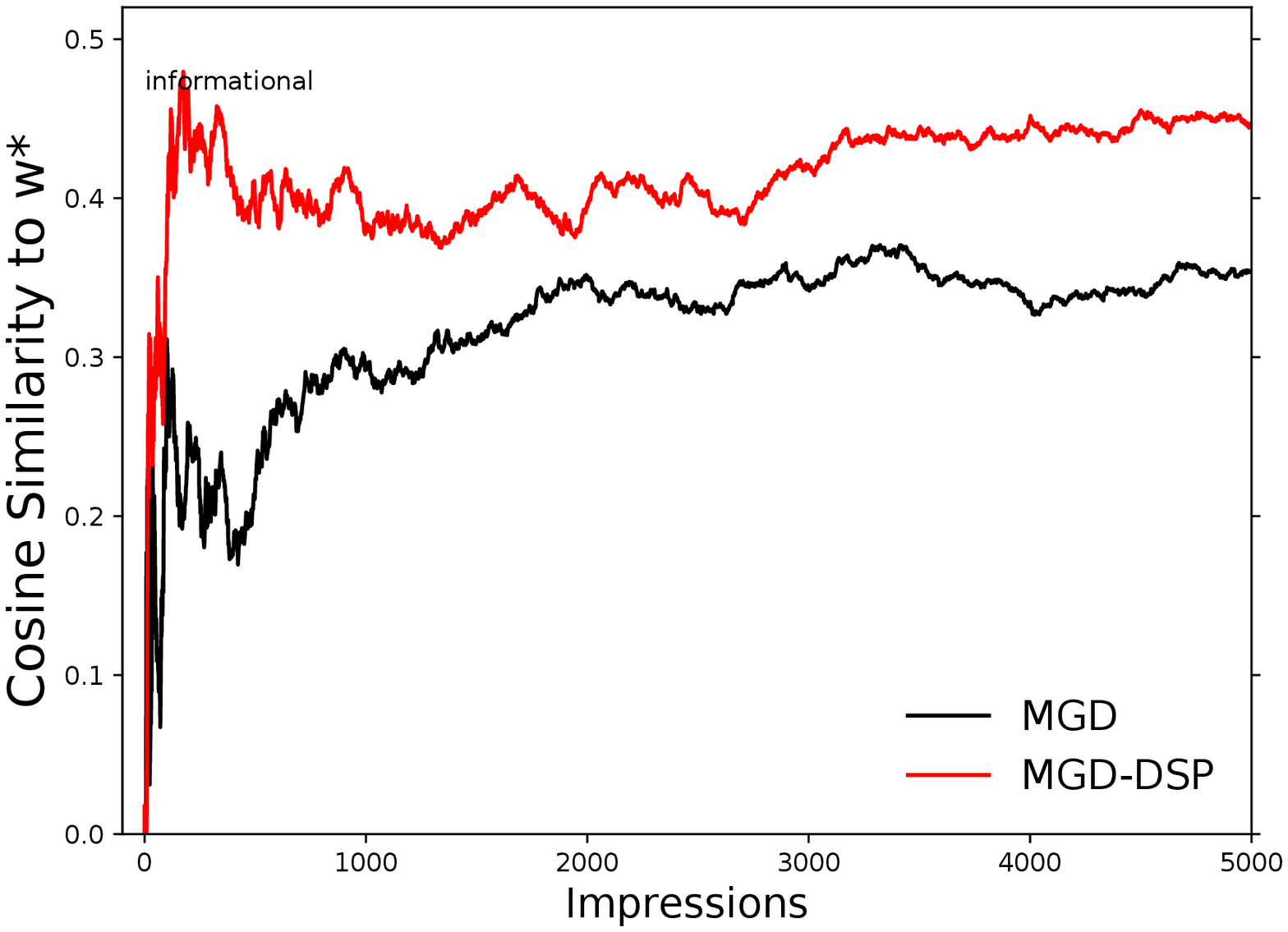} &
\includegraphics[width=4.2cm]{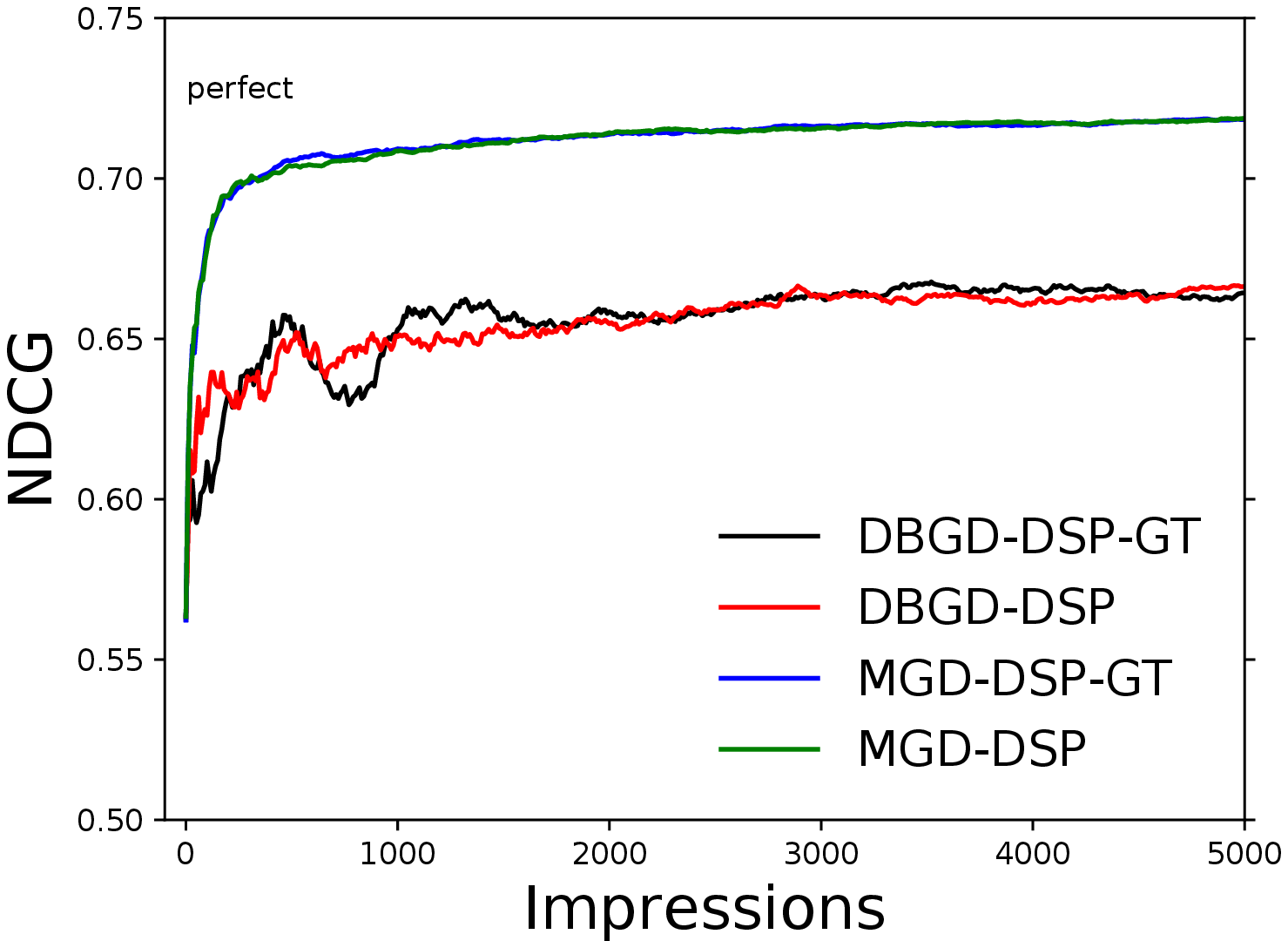}\\
(a) Online performance of linearly interpolating $u_t$ and its projection $g_t$ & (b) Offline performance of linearly interpolating $u_t$ and its projection $g_t$ &(c) Cosine similarity between offline best model $w^*$ and online model &(d) Comparing with ground-truth document space\\
\end{tabular}
\vspace{-2mm}
\caption{Analyzing Document Space Projection.} \label{Fig:analyze}
\vspace{-3mm}
\end{figure*}

\begin{figure}[ht]
\centering
\vspace{-2mm}
\setlength\tabcolsep{4pt}
\begin{tabular}{ >{\centering\arraybackslash}m{4.2cm} >{\centering\arraybackslash}m{4.2cm}}

\hspace*{-3mm}
\includegraphics[width=4.4cm]{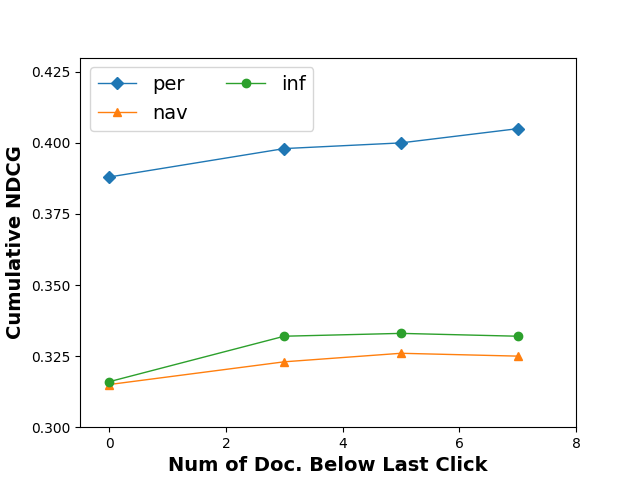} &
\includegraphics[width=4.4cm]{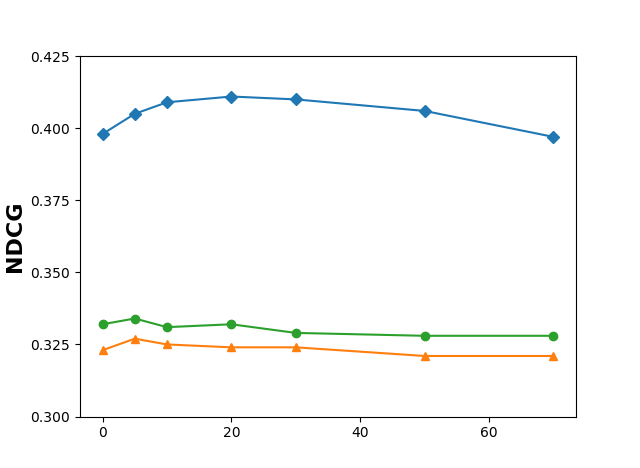} \\
(a) Including $k$ documents following last clicked position &  (b) Including $r$ recently examined documents  \\
\end{tabular}
\vspace{-2mm}
\caption{Hyper-parameter tuning for Document Space Projection.} \label{Fig:param}
\vspace{-4mm}
\end{figure}

\subsection{%Online and offline
Performance of Document Space Projection}
\label{sec_exp_dsp}
We begin our experimental analysis by answering our first evaluation question. We compared all algorithms over 3 click models and 5 datasets. We set the hyper-parameters of DBGD, MGD and NSGD according to their original papers. Following \cite{yue2009interactively,schuth2016multileave}, we set the exploration step size $\delta$ to 1 and learning rate $\alpha$ to 0.1. Both MGD and NSGD explore 9 proposal directions in one iteration. For our document space projection method, we consider $k = 3$ documents following the last clicked position as examined documents, and add $r = 10$ recently examined documents into document space $S_t$. We use SVD to solve for orthonormal basis $V_t$ of the  document space $S_t$, and compute the projection matrix by $A_t = V_t V_t^\top$.

We reported the offline NDCG@10 and online cumulative NDCG @10 after 10,000 iterations in Table \ref{Tab:online} and Table \ref{Tab:offline}. Due to space limit, we only reported the offline performance during the 10,000 iterations over 3 click models on Yahoo dataset, a large-scale real-world L2R dataset with 700 ranking features, in Figure \ref{Fig:offline}. MGD improves the online performance over DBGD by exploring multiple rankers simultaneously, and NSGD further improves over MGD by exploring gradients in a constrained subspace, as shown in Table \ref{Tab:online}. We observe that our proposed document space projection method consistently improves the online performance of all baseline algorithms. Recall that in Section \ref{sec:regret} our theoretical analysis suggested that document space projection reduces both the gradient estimation variance and the regret (online performance) with respect to the ratio between the rank of document space and feature dimension. Correspondingly, we observe that indeed we improved the OL2R models' ranking performance significantly over MSLR-WEB10K and Yahoo datasets, which are collected from real-world commerical search engines and have much higher feature dimensions (130 and 700 respectively). %in terms of relative improvement .
This result demonstrates the potential of document space projection to improve large-scale real-world DBGD-type OL2R applications with high-dimensional ranking features, as our algorithm attains satisfactory performance earlier than other baseline OL2R algorithms measured by online NDCG@10. 
We also notice that the standard deviation of those models' ranking performance is reduced when applying document space projection, which confirms our analysis of variance reduction in Lemma \ref{lemma:variance}.

From Figure \ref{Fig:offline} and Table \ref{Tab:offline} we notice that document space projection mostly improves offline performance over baseline algorithms. Figure \ref{Fig:offline} shows that document space projection significantly accelerates the convergence rate over the baseline algorithms, because of the reduced variance in gradient estimation. We also observe that applying document space projection under the perfect click model may lead to degraded performance, for example DBGD on MQ2007 and Yahoo dataset. This is because document space projection guarantees an unbiased gradient estimation under the assumption of known result examinations, as discussed in Section \ref{sec:unbias}. However, since in practice a user's result examination is unobserved, we approximated the examined documents by including all documents before the last clicked position and $k$ additional documents after the last clicked position. The perfect click model is an ideal case that users' stop probability is set to $0.0$ (see Table \ref{tab:clickmodel}) and every document is examined. Here, the document space needs to include all displayed documents to guarantee the unbiasedness, which requires a significantly larger $k$ compared to the $k$ used for navigational and informational click models.  %but in our evaluation we only displayed the top 10 documents for users to examine.
We argue that in practice since users only examine a handful of documents, we could well-approximate the examined documents with a reasonable choice of $k$. More sophisticated click models can also be introduced. We will analyze the effect of $k$ in Section \ref{sec:analyze}. In addition, we also observe that under informational click model the performance of NSGD-DSP is slightly decreased compared with original NSGD over three datasets. Note that since NSGD does not guarantee its gradient exploration is unbiased, further projecting its gradient may also lead to a biased gradient update and thus a sub-optimal model.

%Conversely,

\subsection{Analysis of Document Space Projection}\label{sec:analyze}
To answer the second evaluation question, we design two experiments to show the effectiveness of document space projected gradient. In the first experiment, we study the utility of document space projected gradient. We compare the ranking performance of linearly interpolating the unrectified direction $u_t$ and its document space projected version $g_t$, i.e., $\lambda g_t + (1-\lambda)u_t$, based on the MGD algorithm on MSLR-WEB10K dataset. Similar observations were obtained on other datasets, but due to space limit we have to omit those detailed results. We report the online and offline performance by varying $\lambda$ from $0$ (which is equivalent to the original MGD algorithm) and $1$ (which is MGD-DSP) in Figure \ref{Fig:analyze} (a) and (b). We can clearly observe a trend of increasing online performance over all three click models when we increase $\lambda$, i.e., trust more on the projected direction $g_t$ for model update. This confirms the effectiveness of the projected direction $g_t$ within document space comparing with the unrectified direction $u_t$ from the entire parameter space. The offline performance is generally robust to the setting of $\lambda$ for navigational and information click models. This is expected since both MGD and MGD-DSP are unbiased and will eventually converge to similar offline performance after sufficiently large number of iterations (we had 10,000 iterations in our experiments).

In the second experiment, we trained an offline LambdaRank model \cite{burges2010ranknet} using the complete annotated relevance labels in the large-scale MSLR-WEB10K dataset. Then given this $w^*$, we compared cosine similarity between the online estimated model parameters with and without DSP in each iteration using MGD as the baseline. We show the result of first 5,000 iterations.  In Figure \ref{Fig:analyze} (c) we can observe that MGD-DSP converges faster and better to $w^*$ than MGD. This suggests the rectified gradient is more effective than the original one. We also compared with an oracle algorithm that knows the ground-truth examined documents, denoted as DSP-GT, to validate the effectiveness of our approximated document space. We show the result on DBGD and MGD under the perfect click model in Figure \ref{Fig:analyze}(d). We notice that oracle algorithms performed similarly as our proposed algorithm with an approximated document space, which confirms the effectiveness of the approximation heuristics.

To answer the third evaluation question, we compare different hyper-parameters used for constructing the document space on MSLR-WEB10K dataset. We vary $k$ from 0 to 7 and report the result in Figure\ref{Fig:param} (a). We notice that for navigational and informational click models, a relatively small $k$ achieved the best performance, i.e., $k=3$. This corresponds to the observation that users do not continue to examine many documents after their last click under these two click models. However, under the perfect click model, the models' performance increases with a larger $k$. This aligns with the conclusions from our discussion in Section \ref{sec_exp_dsp} that under the perfect click model, we need to set a much larger $k$ to accurately construct the document space and guarantee an unbiased gradient estimate.

In Figure \ref{Fig:param}(b), we vary $r$. As we discussed in  Section \ref{sec:practice}, we are motivated to add recently examined documents to compensate for potentially overlooked examined documents in the current query. The effect of different choices of $r$ is more noticeable under the perfect click model. This echoes our analysis above that under the perfect click model some examined documents may be overlooked when $k$ is not large enough. Thus correctly setting up $r$ could reduce the bias in document space construction and compensate the final performance. From the result figure, we notice that setting $r=20$ provides the best result. Under navigational and informational click models, the algorithm is generally robust to the choice of $r$. This is because the approximations of examined documents are already accurate with a reasonable setting of $k$. 

\section{Conclusion}
In this paper, we propose and develop the Document Space Projection (DSP) method for reducing variance in gradient estimation and improving online learning to rank performance. The key insight of DSP is to recognize that the interleaved test only reveals the projection of true gradient on the spanned space of \emph{examined} documents. Including anything beyond this space for model update only introduces noise. Thus our method projects the selected model update direction back to the document space to reduce its variance. We proved that DSP maintains an unbiased gradient estimate, and it can substantially improve the regret bound for DBGD-style algorithms via the reduced variance. Through our extensive experiments, we found that DSP is able to provide statistically significant improvements to several best performing DBGD-type OL2R models, both in terms of variance reduction and overall performance, especially when the number of ranking features is large. 

Currently, we are using a heuristic method to construct the document space. However, we did observe that the performance of DSP varies under different click models for simulated user click feedback, i.e., different underlying examination behaviors. As for our future work, we plan to incorporate different click modeling solutions for more accurate document space construction. It would also be meaningful to study how to perform document space based exploratory direction generation, before the interleaved test. Exploratory direction pre-selection is expected to further accelerate the gradient exploration and improve user satisfaction during online learning, but we also need to ensure it is unbiased. We leave the exploration of our projection-based solution to other types of OL2R algorithms as another future work.

%(Not sure what we want to talk about here, but this is one idea I had based on our discussions ->) In future work, we believe it is critical to study the ethical considerations of online learning to rank. We have identified a possible privacy issue that many OL2R algorithms give rise to: model learning that is dependent on a user's search history and click behavior. Theoretically, one could use these model updates to work backwards and determine a user's actions without their permission. Such situations must be avoided, and we aim to work on producing new, privacy-respecting OL2R methods in the future. Given the prevalence of public privacy concerns in the modern era, this kind of work will undoubtedly be appreciated.
\section*{Acknowledgments}
We thank the anonymous reviewers for their insightful comments. This work was supported in part by National Science Foundation Grant IIS-1553568 and IIS-1618948 and Bloomberg Data Science Ph.D. Fellowship.

\bibliographystyle{ACM-Reference-Format}
\balance
\bibliography{sample-bibliography}

\end{document}